\documentclass[a4paper,UKenglish,cleveref,autoref]{lipics-v2019}

\usepackage[boxed]{algorithm2e}
\usepackage{tikz}
\usepackage{pgfpages}
\usepackage{pgfplots}
\usepackage{pgfplotstable}
\usepackage{float}

\usepackage{csquotes}

\usepackage{xspace}
\usepackage{color}
\newcommand{\poly}{\operatorname{poly}}
\newcommand{\Card}{\operatorname{Card}}
\newcommand{\sdzero}{\textup{\texttt{0}}\xspace}
\newcommand{\sdone}{\textup{\texttt{1}}\xspace}

\newcommand{\classNP}{\text{\rm\textsf{NP}}}
\newcommand{\classFP}{\text{\rm\textsf{FP}}}
\newcommand{\classSharpP}{\text{\rm\textsf{\#P}}}
\newcommand{\classPSPACE}{\text{\rm\textsf{PSPACE}}}

\newcommand{\calC}{\mathcal{C}}
\newcommand{\calO}{\mathcal{O}}
\newcommand{\calU}{\mathcal{U}}
\newcommand{\calSO}{\mathcal{SO}}
\newcommand{\calSU}{\mathcal{SU}}

\newcommand{\IR}{\mathbb{R}}

\newcommand{\IZ}{\mathbb{Z}}
\newcommand{\IN}{\mathbb{N}}
\newcommand{\IH}{\mathbb{H}}

\newcommand{\PM}{\mathcal{PM}}
\newcommand{\Cantor}{{\mathbf{2}^\mathbb{N}}}

\newcommand{\genBall}[2]{\overline{B}_{#1}(#2)}
\newcommand{\abs}[1]{\ensuremath{\left|#1\right|}}

\newcommand{\method}[1]{\texttt{\textup{\ref{proc:#1}}}}

\title{Computing Haar Measures} 

\author{Arno Pauly}{Swansea University, UK \and \url{http://www.cs.swan.ac.uk/~cspauly/}}{arno.m.pauly@gmail.com}{https://orcid.org/0000-0002-0173-3295}{}

\author{Dongseong Seon}{KAIST, Daejeon, Republic of Korea}{instigation@kaist.ac.kr}{}{}

\author{Martin Ziegler}{KAIST, Daejeon, Republic of Korea \and \url{http://ziegler.theoryofcomputation.asia/}}{ziegler@kaist.ac.kr}{}{}

\authorrunning{A.~Pauly, D.~Seon and M.~Ziegler}

\Copyright{Arno Pauly, Dongseong Seon, and Martin Ziegler}



\category{}

\relatedversion{An extended abstract of this work will appear in the proceedings of CSL 2020.}

\supplement{The source code of the implementation in Section \ref{ss:Implement} is available at \linebreak
\href{http://github.com/realcomputation/irramplus/tree/master/HAAR}{\texttt{https://github.com/realcomputation/irramplus/tree/master/HAAR}}}

\funding{Supported by the National Research Foundation of Korea
(grant NRF-2017R1E
1A1A03071032),
by the International Research \& Development Program of
the Korean Ministry of Science and ICT (grant NRF-2016K1A3A7A03950702),
and by the \protect\raisebox{-1pt}{\protect\includegraphics[height=8pt]{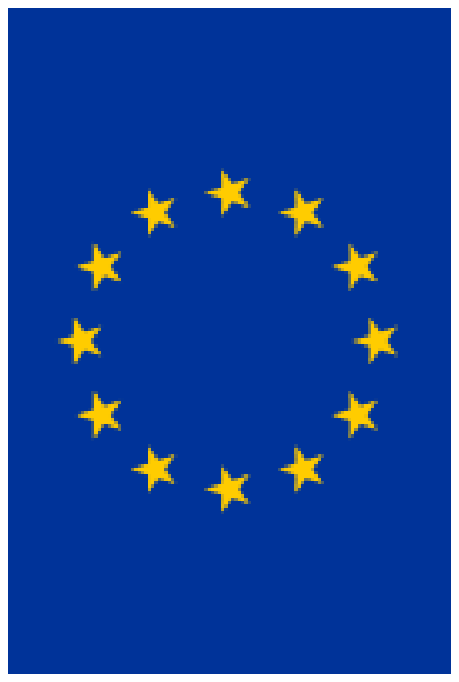}}~European Union's Horizon 2020 MSCA IRSES project 731143.}

\nolinenumbers 

\hideLIPIcs  

\EventEditors{Maribel Fern\'{a}ndez and Anca Muscholl}
\EventNoEds{2}
\EventLongTitle{28th EACSL Annual Conference on Computer Science Logic (CSL 2020)}
\EventShortTitle{CSL 2020}
\EventAcronym{CSL}
\EventYear{2020}
\EventDate{January 13--16, 2020}
\EventLocation{Barcelona, Spain}
\EventLogo{}
\SeriesVolume{152}
\ArticleNo{34}

\theoremstyle{definition}
\newtheorem*{theorem*}{Theorem}
\newtheorem{observation}[theorem]{Observation}
\newtheorem{fact}[theorem]{Fact}

\begin{document}

\maketitle

\begin{abstract}
According to Haar's Theorem, every compact topological group $G$
admits a unique (regular, right and) left-invariant Borel probability measure $\mu_G$.
Let the \emph{Haar integral} (of $G$) denote the functional
$\int_G:\calC(G)\ni f\mapsto \int f\,d\mu_G$
integrating any continuous function $f:G\to\IR$ with respect to $\mu_G$.
This generalizes, and recovers
for the additive group $G=[0;1)\mod 1$,
the usual Riemann integral: computable (cmp. Weihrauch 2000, Theorem 6.4.1),
and of computational cost characterizing complexity class \classSharpP$_1$
(cmp. Ko 1991, Theorem 5.32).

We establish that in fact, every computably compact computable metric group
renders the Haar measure/integral computable: 
once asserting computability using an elegant synthetic argument,
exploiting uniqueness in a computably compact space of probability measures;
and once presenting and analyzing an explicit, imperative algorithm
based on `maximum packings' with rigorous error bounds and guaranteed convergence.
Regarding computational complexity, for the groups $\calSO(3)$ and $\calSU(2)$,
we reduce the Haar integral to and from Euclidean/Riemann integration.
In particular both also characterize \classSharpP$_1$.
Implementation and empirical evaluation using the iRRAM C++ library for exact real computation
confirms the (thus necessary) exponential runtime.
\end{abstract}

\section{Motivation and Overview}
\label{sec:Overview}
Complementing empirical approaches, heuristics, and recipes \cite{PAUSINGER201913,Recipes},
Computable Analysis \cite{Wei00} provides a rigorous algorithmic foundation to Numerics,
as well as a way of formally measuring the constructive contents of
theorems in classical Calculus.
Haar's Theorem is such an example, of particular beauty combining three categories:
compact metric spaces, algebraic groups, and measure spaces:

\begin{fact} \label{f:Haar}
Let $(G,e,\circ,\cdot^{-1})$ denote a group
and $(G,d)$ a compact metric space
such that the group operation $\circ$ and inverse operation $\cdot^{-1}$ are continuous with respect to $d$ (that is, form a topological group).
There exists a unique left-invariant Borel probability measure $\mu_G$, called Haar measure, on $G$.
Moreover, $\mu_G$ is right-invariant and regular.
\end{fact}
We refrain from expanding on generalizations to locally-compact Hausdorff spaces.
Recall that 
a left-invariant measure satisfies $\mu(U)=\mu(g\circ U)$ for every $g\in G$ and every measurable $U\subseteq G$.
For the additive group $[0;1)\mod 1$, its Haar measure recovers the standard Lebesgue measure $\lambda$,
corresponding to the angular measure divided by $2\pi$
on the complex unit circle group $\calU(1)\cong\calSO(2)$.

Each of the categories involved in Fact~\ref{f:Haar}
has a standard computable strengthening, cmp. \cite{Wei93,SS06,Eda09};
and our first main result establishes them to combine nicely:

\begin{theorem}
\label{theo:synthetic:nice}
Let $\mathbf{X}$ be a computably compact computable metric space with a computable group operation $\mathalpha{\circ} : \mathbf{X} \times \mathbf{X} \to \mathbf{X}$.
Then the corresponding Haar measure $\mu$ is computable.
\end{theorem}
That the Haar measure is computable means that we can approximate the measure of any given open subset of $\mathbf{X}$ from below, and implies that we can compute the integral of any given continuous function from $\mathbf{X}$ into $\mathbb{R}$.

Note that Haar's theorem would usually be stated for (locally)-compact Hausdorff topological groups, instead of requiring the metric structure. From the perspective of computable analysis, this generality is ephemeral, however. The topological spaces with a matching notion of computability are the $\mathrm{QCB}_0$-spaces \cite{schroder}, and any locally compact $\mathrm{QCB}_0$-space is already countably based (\cite[Corollary 6.11]{escardo7}), and countably-based compact Hausdorff spaces are metrizable.

In contrast, recall that other classical results in Calculus,
such as Brouwer's Fixed Point Theorem \cite{Ore63,paulybrattka3} or Peano's Theorem \cite{PR79},
do not carry over to computability that nicely.
And also common classical `constructive' existence proofs of the Haar measure
\cite[\S58]{HalmosGTM18} do employ limits without rate of convergence,
well-known since Specker \cite{Spe49,Spe59} to possibly leave the computable realm:

\begin{fact}
\label{f:Classical}
For non-empty $A,B\subseteq X$ let $[A:B]$ denote the least number of
left translates of $B$ that cover $A$.
Then $\mu(A)=\lim_B \frac{[A:B]}{[X:B]}$
holds for every compact $A\subseteq X$,
where the limit exists in the sense
of a net of open neighborhoods $B$ of $e$.
\end{fact}
In addition to the possibly uncomputable limit,
the \emph{least integer} defining $[A:B]$ depends discontinuously and uncomputably on the underlying data $A,B$.

We establish Theorem~\ref{theo:synthetic:nice}
with elegant arguments following the `synthetic' (i.e. implicit, functional) approach to Computable Analysis developed in \cite{Pau16}.
It follows the following general strategy \cite{Esc13} (also explained in \cite[Section~9]{Pau16})
for proving computability of some object $\Omega$ living in an admissibly represented space by three steps:
\begin{enumerate}
\item[I)] Obtain a definition of $\Omega$ as the element of a computably closed set.
\item[II)] Obtain a computably compact set containing $\Omega$. 
\item[III)] Find a classical proof that (I) and (II) uniquely determine $\Omega$.
\end{enumerate}
As warm-up let us illustrate this approach to assert
computability of the group unit $e$ 
from the hypothesis of \cref{theo:synthetic:nice}:
For any fixed computable element $a\in G$,
(I) $e$ belongs to the computably closed set $\{y: a\circ y=a\}$
and (II) to the compact set $\Omega=X$
and (III) is uniquely determined by (I) and (II).
Similarly, for every $x\in G$, its inverse $x^{-1}$
(I) belongs to $\{y: x\circ y=e\}$
and (II) to the compact set $\Omega=X$
and (III) is uniquely defined by $x\circ x^{-1}=e$.
Note that this proof does not immediately yield an
algorithm computing $e$ or $x\mapsto x^{-1}$.

In this spirit, Section~\ref{sec:synthetic} establishes
Theorem~\ref{theo:synthetic:nice}. The challenge consists in
(I) obtaining a computable definition of the Haar measure $\mu$:
The inequality $\tilde\mu(U)\neq \tilde\mu(xU)$ expressing
a candidate measure $\tilde u$ to violate invariance
is not even recognizable, since $\tilde\mu(U)$ is in general
only a lower real. Subsection~\ref{ss:Pairs} avoids that
by allowing to consider pairs of sets in Lemma~\ref{lemma:allpairinvarianceclosed}.
Section~\ref{sec:Dongseong} complements Section~\ref{sec:synthetic}
by devising and analyzing an explicit, imperative algorithm
for computing Haar integrals $\calC(f)\mapsto\int_X f\,d\mu$.
It is based on `maximal packings': finite sets $T_n\subseteq X$ of points with pairwise distance $>2^{-n}$.
Intuitively, the ratio $|T_n\cap A|/|T_n|$ of those points contained in a given set $A$
should approximate its measure $\mu(A)$; however, rigorously, this is wrong
--- and counting is uncomputable anyway. Subsections~\ref{ss:Estimates}
and \ref{ss:Algorithm} describe a combination of mathematical and algorithmic
approaches that avoid these obstacles.
The superficially different hypotheses to Sections~\ref{sec:Dongseong}
and \ref{sec:synthetic} are compared in Section~\ref{sec:compare}.
There we also give some examples showing that these requirements are not dispensable;
and analyze which information of a compact metric group determines its Haar measure.

Having thus asserted computability, the natural next question is for efficiency.
We consider here the non-uniform computational cost of the
\emph{Haar integral} functional
\begin{equation}
\label{e:Average}
\int\nolimits_G \;:\; \calC(G)\;\ni\; f\;\mapsto\; \int f\,d\mu_G  \;\in\; \IR
\end{equation}
integrating continuous real functions $f:G\to\IR$.
For the arguably most important additive groups $G=[0;1)^d$ mod 1 with Lebesgue measure $\lambda^d$,
this amounts to Euclidean/Riemann integration --- whose complexity
had been shown to characterize the discrete class $\classSharpP_1$
\cite[Theorem~5.32]{Ko91} cmp. \cite{DBLP:conf/macis/FereeZ15,lmcs:3924}: indicating that standard quadrature methods,
although taking runtime exponential in $n$ to achieve guaranteed absolute
output error $2^{-n}$, are likely optimal.
And Section~\ref{sec:Complexity} extends this numerical characterization
of $\classSharpP_1$ to the arguably next-most important compact metric groups:

\begin{theorem}
\label{theo:Complexity}
Let $G$ denote any of the following compact groups,
considered as subsets of Euclidean space and equipped with the intrinsic/path metric:
\begin{enumerate}
\item[i)] $\calSO(3)\subseteq\IR^9$ of orthogonal real $3\times 3$ matrices of determinant $1$,
\item[ii)] $\calO(3)\subseteq\IR^9$ of orthogonal real $3\times 3$ matrices,
\item[iii)] $\calSU(2)\subseteq\IR^8$ of unitary complex $2\times 2$ matrices of determinant $1$,
\item[iv)] $\calU(2)\subseteq\IR^8$ of unitary complex $2\times 2$ matrices.
\smallskip\smallskip
\item[a)]
For every polynomial-time computable $f\in\calC(G)$,
$\int_G f\in\IR$ is computable in polynomial space
(and exponential time).
\item[b)]
If $\classFP_1=\classSharpP_1$ and
$f\in\calC(G)$ is polynomial-time computable,
then so is $\int_G f\in\IR$.
\item[c)]
There exists a polynomial-time computable $f\in\calC(G)$
such that polynomial-time computability of $\int_G f\in\IR$
implies $\classFP_1=\classSharpP_1$.
\end{enumerate}
\end{theorem}
The proof of this result proceeds by mutual polynomial-time continuous (i.e. Weihrauch) reduction
from and to Euclidean/Riemann integration.
Subsection~\ref{ss:Implement} describes our implementation and empirical evaluation
of rigorous integration on $\calSU(2)$ in the \texttt{iRRAM} C++ library.

\section{Background}
\label{sec:Background}

In the following, we give a brief introduction to the key notions from computable analysis we need. For a formal treatment, we refer to \cite{Pau16}. Further standard references for Computable Analysis are \cite{Wei00,BHW08}.

Computable analysis is concerned with \emph{represented spaces}, which equip a set with a notion of computability by coding its elements as infinite binary sequences. We have various constructions of new represented spaces available, and use in particular the derived spaces $\mathcal{O}(\mathbf{X})$ of open subsets of $\mathbf{X}$ (characterized by making membership recognizable) and $\mathcal{C}(\mathbf{X},\mathbf{Y})$ of continuous functions from $\mathbf{X}$ to $\mathbf{Y}$, characterized by making function evaluation computable.

Computable compactness and computable overtness of a space are characterized by making universal and existential quantification preserve computable open predicates. We also use that admissibility of a space means that from a compact singleton we can extract the point \cite{Sch06}. A space is computably Hausdorff, if inequality is semidecidable. It is computably separable, if it has a computable dense sequence. Being computably separable implies being computably overt.

A particular convenient class of represented space are the computable metric spaces (CMS). We can start with a designated dense sequence on which the metric is computable (given indices), and then represent arbitrary points as limits of fast converging sequences. CMSs are in particular computably Hausdorff and computably separable. The prototypic example of a CMS are the reals $\mathbb{R}$. We write \emph{CCCMS} for \emph{computably compact computable metric space}.

There is a further relevant represented space with the reals as underlying set, namely the space of lower reals $\mathbb{R}$. Here a real is represented as the supremum of a sequence of rational numbers (without any limitation on convergence rates). This space is relevant for us as we can introduce the computability structure of the space of (probability) measures on an arbitrary represented space $\mathbf{X}$ by considering them as the subspace of $\mathcal{C}(\mathcal{O}(\mathbf{X}),\mathbb{R}_<)$ of functions satisfying the properties of a (probability) measure. More precisely, these correspond to continuous valuations. Let $\PM(\mathbf{X})$ denotes the space of probability measures on $\mathbf{X}$.

A useful theorem is that for a CCCMS $\mathbf{X}$, also $\PM(\mathbf{X})$ is a CCCMS. Here we can use the \emph{Wasserstein-Kantorovich-Rubinstein} metric
\[ W(\mu,\nu) \;=\; \sup\Big\{ \big|\int f\,d\mu \:-\: \int f\,d\nu \big| \;:\; f:X\to\IR, \; \forall x,y\in X: \; |f(x)-f(y)|\leq d(x,y) \Big\} \]
If $\mathbf{X}$ a complete metric space, $\PM(\mathbf{X})$ is again a complete metric space;
and convergence w.r.t. $W$ is equivalent to weak convergence.
For an introduction to computable probability theory, see \cite{collins4}. Some further results are found in \cite{pauly-fouche}.

Regarding computational complexity of real numbers and real functions on compact metric spaces, we refer to \cite{Ko91} and \cite{DBLP:conf/lics/KawamuraS016}.

Recall that $\classSharpP_1$ is the class of all integer functions $\varphi:\{\sdzero\}^*\to\IN$
with unary arguments counting the number of witnesses
\[ \varphi(\sdzero^n) \;=\; \Card\big\{\vec w\in\{\sdzero,\sdone\}^{\poly(n)} : \sdzero^n\,\sdone\,\vec w\in P\big\} \]
to a polynomial-time decidable predicate $P\subseteq\{\sdzero,\sdone\}^*$;
a class commonly conjectured to lie strictly between (the integer function versions of)
$\classNP_1$ and $\classPSPACE$ \cite[\S18]{Papadimitriou}.

\section{The Haar measure is computable}
\label{sec:synthetic}
In this section we shall establish Theorem~\ref{theo:synthetic:nice}
using the approach to computable analysis via synthetic topology \cite{Esc04} outlined in \cite{Pau16}.
To this end, we first obtain a more technical result stating that left-invariance of a
Radon probability measure for some continuous binary operation constitutes a computably closed predicate:

\begin{theorem}
\label{theo:synthetic:main}
Let $\mathbf{X}$ be a computable metric space. For $\mu \in \PM(\mathbf{X})$ and $g \in \mathcal{C}(\mathbf{X} \times \mathbf{X},\mathbf{X})$ the following predicate is computably closed:
\[\forall U \in \mathcal{O}(\mathbf{X}), \ \forall x \in \mathbf{X} \ \mu(U) = \mu(\{y \in \mathbf{X} \ g(x,y) \in U\})\]
\end{theorem}

In view of the general strategy for computability proofs from Section~\ref{sec:Overview}, this establishes (I).
Regarding (II) recall \cite[\S2.5]{GHR11} that, if $\mathbf{X}$ is a computably compact computable metric space, then so is $\PM(\mathbf{X})$.
Finally, uniqueness in Haar's theorem takes care of Condition~(III).

\subsection{Disjoint pairs of open sets}
\label{ss:Pairs}
Prima facie, the condition in Theorem \ref{theo:synthetic:main} appears to be complicated. As measures of open sets are only available as lower reals, we cannot even recognize inequality. The workaround consists in considering pairs of disjoint open sets rather than individual open sets. We shall see that quantification over such pairs is unproblematic for the spaces we are interested in here.

Given a represented space $\mathbf{X}$, we define the space $\mathrm{DPO}(\mathbf{X})$ as the subspace $\{(U,V) \mid U \cap V = \emptyset\} \subseteq \mathcal{O}(\mathbf{X}) \times \mathcal{O}(\mathbf{X})$.

\begin{observation}
\label{obs:dpoclosed}
$\mathbf{X}$ is computably overt iff $\mathrm{DPO}(\mathbf{X})$ is a computable element of $\mathcal{A}(\mathcal{O}(\mathbf{X}) \times \mathcal{O}(\mathbf{X}))$.
\begin{proof}
If $\mathbf{X}$ is computably overt, then $U \cap V \neq \emptyset$ is a recognizable property given $(U,V) \in \mathcal{O}(\mathbf{X}) \times \mathcal{O}(\mathbf{X})$. Conversely, we find that $(U, X) \notin \mathrm{DPO}(\mathbf{X})$ iff $U \neq \emptyset$.
\end{proof}
\end{observation}

\begin{corollary}
If $\mathbf{X}$ is computably overt, then $\mathrm{DPO}(\mathbf{X})$ is computably compact.
\begin{proof}
The space $\mathcal{O}(\mathbf{X})$ is computably compact, as it contains $\emptyset$ as a computable bottom element. Then $\mathcal{O}(\mathbf{X}) \times \mathcal{O}(\mathbf{X})$ is computably compact as a product, and finally the claim follows by noting that a computably closed subspace of a computably compact space is computably compact and invoking Observation \ref{obs:dpoclosed}.
\end{proof}
\end{corollary}

\begin{lemma}
\label{lemma:dpoisovert}
If $\mathbf{X}$ is computably separable, effectively countably based and computably Hausdorff, then $\mathrm{DPO}(\mathbf{X})$ is a computable element of $\mathcal{V}(\mathcal{O}(\mathbf{X}) \times \mathcal{O}(\mathbf{X}))$.
\begin{proof}
It is shown in \cite{paulytsuiki-arxiv} that under the given conditions, we can obtain an adequate formal disjointness notion on basic open sets. We can then obtain a dense sequence in $\mathrm{DPO}(\mathbf{X})$ by constructing pairs of finite unions of basic open sets with the additional requirements that each basic open set is formally disjoint from all basic open sets listed in the opposite finite union.
\end{proof}
\end{lemma}

\begin{corollary}
\label{corr:dpoisnice}
If $\mathbf{X}$ is computably separable, effectively countably based and computably Hausdorff, then $\mathrm{DPO}(\mathbf{X})$ is computably compact and computably overt.
\end{corollary}

\begin{definition}
Given $f \in \mathcal{C}(\mathbf{X},\mathbf{X})$, $(U,V) \in \mathrm{DPO}(\mathbf{X})$ and $\mu \in \PM(\mathbf{X})$, we say that $(U,V)$ is $\mu$-invariant under $f$, iff: $$\mu(U) + \mu(f^{-1}(V)) \leq 1$$
\end{definition}

\begin{observation}
\label{obs:invarianceclosedforsinglepair}
$(U,V)$ being $\mu$-invariant under $f$ is a computably closed property.
\end{observation}

\begin{lemma}
\label{lemma:allpairinvarianceclosed}
Let $\mathbf{X}$ be computably separable, effectively countably based and computably Hausdorff. Then ``\emph{all pairs from $\mathrm{DPO}(\mathbf{X})$ are $\mu$-invariant under $f$}'' is a computably closed property in $\mu \in \PM(\mathbf{X})$ and $f \in \mathcal{C}(\mathbf{X},\mathbf{X})$.
\begin{proof}
Computably closed properties are closed under universal quantification over computably overt sets. So we just combine Observation \ref{obs:invarianceclosedforsinglepair} and Corollary \ref{corr:dpoisnice}.
\end{proof}
\end{lemma}

\subsection{Proof of Theorem \ref{theo:synthetic:main}}
To be able to invoke the results of the previous subsection we need to relate invariance of disjoint pairs of open sets to invariance of individual open sets.

\begin{lemma}
\label{lem:setsvspairs}
For a computable metric space $\mathbf{X}$,  $\mu \in \PM(\mathbf{X})$ and $f \in \mathcal{C}(\mathbf{X},\mathbf{X})$ the following are equivalent:
\begin{enumerate}
\item All pairs from $\mathrm{DPO}(\mathbf{X})$ are $\mu$-invariant under $f$.
\item For all $U \in \mathcal{O}(\mathbf{X})$ it holds that $\mu(U) = \mu(f^{-1}(U))$.
\end{enumerate}
\end{lemma}

\begin{proof}
2. implies
\[ \mu(U)\;+\;\mu\big(f^{-1}(V)\big)
\;=\;
\mu\big(f^{-1}(U)\big)\;+\;\mu\big(f^{-1}(V)\big)
\overset{(*)}{\;=\;}
\mu\big(f^{-1}(U)\cup f^{-1}(V)\big)  \;\leq\; 1 \]
with (*) since $f^{-1}(U)$ and $f^{-1}(V)$ are disjoint.
\\

For the converse, assume that $U$ witnesses that $f$ is not invariant, i.e.~$\mu(U) \neq \mu(f^{-1}(U))$. We shall argue that this implies the existence of a disjoint pair of open sets which is not $\mu$-invariant under $f$. Let $\delta = \frac{1}{3} |\mu(U) - \mu(f^{-1}(U))|$. Consider the sets $B_{-\varepsilon}(U) = \{x \in \mathbf{X} \mid d(x,U^C) > \varepsilon\}$. Since $U = \bigcup_{\varepsilon > 0} B_{-\varepsilon}(U)$ is a nested union and $f$ is continuous, we find that $\mu(U) = \sup_{\varepsilon > 0} \mu(B_{-\varepsilon}(U))$ and $\mu(f^{-1}(U)) = \sup_{\varepsilon > 0} \mu(f^{-1}(B_{-\varepsilon}(U)))$. Consequently, there exists some $\varepsilon_0$ such that for all $\varepsilon < \varepsilon_0$ it holds that $|\mu(U) -  \mu(B_{-\varepsilon}(U))| < \delta$ and $|\mu(f^{-1}(U)) -  \mu({f^{-1}(B_{-\varepsilon}(U))})| < \delta$.

Next, consider the sets $D_{-\varepsilon}(U) := \{x \in \mathbf{X} \mid d(x,U^C) = \varepsilon\}$. Since for different $\varepsilon$ these sets are disjoint, we know that for only countably many $\varepsilon$ can it hold that $\mu(D_{-\varepsilon}(U)) > 0$. The sets $f^{-1}(D_{-\varepsilon}(U))$ are disjoint, too, and thus the same argument applies. We can thus select some $\varepsilon_1 < \varepsilon_0$ such that $\mu(D_{-\varepsilon_1}(U)) = \mu(f^{-1}(D_{-\varepsilon_1}(U))) = 0$. This ensures that $\mu(B_{-\varepsilon_1}(U)) + \mu((B_{-\varepsilon_1}(U)^C)^\circ) = 1$ and $\mu(f^{-1}(B_{-\varepsilon_1}(U))) + \mu(f^{-1}((B_{-\varepsilon_1}(U))^C)^\circ) = 1$. Moreover, we know that $|\mu(B_{-\varepsilon_1}(U)) - \mu(f^{-1}(B_{-\varepsilon_1}(U)))| > \delta$ from $\varepsilon_1 < \varepsilon_0$, so depending on the sign of the difference, either $(B_{-\varepsilon_1}(U), (B_{-\varepsilon_1}(U)^C)^\circ)$ or $((B_{-\varepsilon_1}(U)^C)^\circ, B_{-\varepsilon_1}(U))$ is not $\mu$-invariant under $f$.
\end{proof}

\begin{proof}[Proof of Theorem \ref{theo:synthetic:main}]
By \cref{lem:setsvspairs} we can replace the invariance for open sets by invariance for disjoint pairs of open sets. By Lemma \ref{lemma:allpairinvarianceclosed}, this is a computably closed property for each fixed choice of continuous function $y \mapsto g(x,y)$. The additional universal quantification over the computably overt space $\mathbf{X}$ preserves being a computably closed predicate.
\end{proof}

\section{Explicit computation of the Haar measure}
\label{sec:Dongseong}

The synthetic arguments from Section~\ref{sec:synthetic}
establishing computability (Theorem~\ref{theo:synthetic:nice})
do not immediately exhibit an actual algorithm.
To this end, the present section takes a more explicit approach.
Its assumptions superficially differ but will be shown equivalent
(in a sense to be formalized) in \cref{sec:compare}.
Among others, we suppose computability of the \emph{size of maximum packings}.
This is a notion asymptotically related to, yet in detail (maximum packing vs. minimum covering,
open vs. closed balls) subtly different from, Kolmogorov's metric entropy \cite{Kolmogorov},
to the \emph{separation bound} from \cite[Definition~6.2]{Wei03}, and to the \emph{capacity} from \cite[Definition~12]{DBLP:conf/lics/KawamuraS016}.
All three notions can be regarded as integer Skolemizations (i.e. \emph{moduli}) of total boundedness \cite[Def~17.106]{Koh08a}.

\begin{definition}
\label{d:Packing}
For any compact metric space $(X, d)$ and its subset $U \subseteq X$,
    \begin{enumerate}
        \item $T \subseteq U$ is called an \emph{$n$-packing} of $U$ if $\forall x,y \in T (x \neq y) \rightarrow d(x,y) > 2^{-n}$.
        \item An $n$-packing $T$ is \emph{maximum} if $|T|\geq |S|$ for every $n$-packing $S$ of $U$.
        \item $\{T_n\}_{n=1}^\infty$ is a \emph{sequence} of maximum packings if each $T_n$ is a maximum $n$-packing.
        \item $\kappa_U: \IN \rightarrow \IN$ is the \emph{size} of maximum packings of $U$ if $\kappa_U(n) =|T_n|$ where $T_n$ is a maximum $n$-packing.
    \end{enumerate}
If $U = X$, the term `of $U$' is omitted.
\end{definition}
Our definition features strict inequality of pairwise distances:
this implies that a maximum $n$-packing $T_n$ can be found algorithmically
by exhaustive search, provided that its size is given/computable:

\begin{fact}
\label{f:Packing}
    For any computable metric space $(M,d,A,\alpha)$, if the size of maximum packings $\kappa_X$ is computable, then for any $n$, we can compute an encoding of a maximum $n$-packing $T_n$ in $A$. In other words, we can compute $S \subseteq \Sigma^*$ s.t. $\alpha(S)$ is a maximum $n$-packing.
\end{fact}
\begin{proof}
    First, let's show that $A$ includes a maximum $n$-packing as a subset. Fix a maximum $n$-packing $T_n$ and let $T_n := \{p_1, \dots, p_N\}$. Since $\forall i,j (i \neq j) \Rightarrow d(p_i,p_j) > 2^{-n}$, if we let $R := \min_{i \neq j} d(p_i, p_j)$ then $R>2^{-n}$. Let $\delta = R-2^{-n}$. Then $\forall p_i \in T_n, \exists {p_i}' \in A$ s.t. ${p_i}' \in B_{\frac{\delta}{2}}(p_i)$ since $A$ is dense. Thus, $\{p_1', \dots, p_N'\} \subseteq A$ and $\forall i,j (i \neq j) \Rightarrow d(p_i', p_j') > d(p_i,p_j) - d(p_i,p_i') - d(p_j,p_j') > R-\delta \geq 2^{-n}$, which means $\{p_1', \dots, p_N'\}$ is a maximum $n$-packing.
    \\Second, let's show that there is an algorithm that outputs a maximum $n$-packing in $A$ if there is one. The algorithm is to dovetail the test of distance between $\kappa_X(n)$ element subsets of $A$. since the test does not includes equality, it is semidecidable. Thus, the algorithm will output a maximum $n$-packing if there is one.\\
    Combining the first and the second step gives the computability of a maximum $n$-packing.
\end{proof}

\begin{theorem}
\label{theo:explicit}
Let $(X, d)$ be a computable metric space and $(X,e,\circ,\cdot^{-1})$ a compact topological group.
Suppose that the metric $d$ is bi-invariant:
\[ \forall a,b,c\in X: \; d(a \circ c, b \circ c) = d(a,b) = d(c \circ a, c \circ b) \]
And suppose that the size of maximum packings $\kappa_X:\IN\to\IN$ is computable.
Then the Haar integral $\calC(X)\ni f\mapsto\int_X f\,d\mu$ is computable.
\end{theorem}
Recall \cite[\S8.1]{Wei00} that a computable metric space $(X,d)$
comes with a dense sequence $\xi:\IN\to X$ such that the
real double sequence $d:\IN\times\IN\ni (a,b)\mapsto d\big(\xi(a),\xi(b)\big)$ is computable.
Note that, as opposed to \cref{theo:synthetic:nice},
we do not suppose the group operation $\circ$ (nor neutral element nor inversion)
to be computable but instead require the metric to be bi-invariant.
See Section~\ref{sec:compare} for a comparison
between the different hypotheses.

\subsection{Mathematical Estimates of Haar Measures}
\label{ss:Estimates}
Invariance of both metric $d$ and Haar measure $\mu$ implies that the content $\mu(B)$ of an open ball $B=B_r(c)$ depends only on its radius $r$, but not on its center $c$.
Intuitively, for a sufficiently large maximum packing $T$, said volume should be approximated by the ratio of points in $B$ to the total number of points
(Definition~\ref{d:Counting}).
If $B_r(c)$ contains significantly smaller a fraction, then by double counting some other $B_r(c')$ would need to `compensate' with a larger fraction,
hence invariance suggests that more points can be added to $T$ at $B(r,c)$ as well, contradicting maximality.
Lemma~\ref{lem:evenDist} below formalizes this idea both in its statement and proof. 
\begin{definition}
    For a metric space $(X, d)$ and its subset $U \subseteq X$, we introduce the outer generalized closed ball as $\genBall{r}{U} := \{x \in X \mid d(x,U) \leq r\}$.
    Similarly, the inner generalized closed ball is introduced as $\genBall{-r}{U} = \{ x\in X : d(x,U^c) \geq r\}$.
\end{definition}
For $0\leq r,s$ it holds
\begin{equation}
\label{e:Balls}
\genBall{+r}{\genBall{-r}{U}} \;\subseteq\; \overline{U} \;\subseteq\;
\genBall{-r}{\genBall{+r}{U}}, \qquad
\genBall{+r}{\genBall{+s}{U}} \;\subseteq \; \genBall{+r+s}{U}
\end{equation}

\begin{lemma}\label{lem:evenDist}
    Suppose $(X,d, \circ)$ is a compact topological group with bi-invariant metric $d$ and a maximum $n$-packing $T_n$ of size $\kappa_X(n)$.
        Then for any $x \in X$ and measurable $U \subseteq X$ it holds:
    $$\kappa_{\genBall{-2^{-n}}{U}}(n) \leq
    |T_n \cap xU| \leq
    \kappa_{\genBall{2^{-n}}{U}}(n)$$
\end{lemma}
\begin{proof}
    $\abs{T_n \cap xU} \leq \kappa_{xU}(n)$ since $T_n \cap xU$ is an $n$-packing of $xU$. Because of bi-invariance, $\kappa_U(n) = \kappa_{xU}(n)$. So $\abs{T_n \cap xU} \leq \kappa_U(n) \leq \kappa_{\genBall{2^{-n}}{U}}(n)$ since $U \subseteq \genBall{2^{-n}}{U}$.\\
    Let $S_n$ be a maximum $n$-packing of $\genBall{-2^{-n}}{xU}$. Then $S_n \cup (T_n \cap (xU)^c)$ is an $n$-packing for the whole space $X$. Since $T_n$ is maximum, $\abs{S_n} + \abs{T_n \cap (xU)^c} = \abs{S_n \cup (T_n \cap (xU)^c)} \leq \abs{T_n}$. Thus $\abs{T_n \cap xU} \geq \abs{S_n} = \kappa_{\genBall{-2^{-n}}{U}}(n) = \kappa_{\genBall{-2^{-n}}{xU}}(n)$.\\
\end{proof}

\begin{definition}
\label{d:Counting}
Abbreviate
$\mu_T := \frac{1}{|T|} \sum_{p \in T} \delta_p$
where $\delta_p$ denotes the Dirac measure.
\end{definition}

\begin{lemma}\label{lem:coreIneq}
Let $(X,d, \circ)$ be a compact topological group with bi-invariant metric $d$ and Haar measure $\mu$,
and $T_n$ a maximum $n$-packing. Then for any $U \subseteq X$:
$$\mu\big(\genBall{-2^{-n+2}}{U}\big) \;\leq\;
\mu_{T_n}\big(\genBall{-2^{-n+1}}{U}\big) \;\leq\;
\mu(U) \;\leq\;
\mu_{T_n}\big(\genBall{2^{-n+1}}{U}\big) \;\leq\;
\mu\big(\genBall{2^{-n+2}}{U}\big)$$
\end{lemma}
\begin{proof}
Let $T_n = \{p_1, \dots, p_N\}$ so that $\abs{T_n} = N$. Then $\mu(U)N = \sum_{i=1}^N \int_X \chi_{p_iU} d\mu = \int_X \sum_{i=1}^N \chi_{p_iU} d\mu$. Let $f := \sum_{i=1}^N \chi_{p_iU}$ then $f(x) = |\{ i : x \in p_iU \}| = |\{ i : p_i \in xU^{-1} \}| = |T_n \cap xU^{-1}|$. By \cref{lem:evenDist}, $f(x) \leq \kappa_{\genBall{2^{-n}}{U^{-1}}} \leq |T_n \cap \genBall{2^{-n+1}}{U^{-1}}|$. Dividing both sides of $\mu(U)N = \int_X f d\mu \leq |T_n \cap \genBall{2^{-n+1}}{U^{-1}}|$ by $N$ gives $\mu(U) \leq  \mu_{T_n}(\genBall{2^{-n+1}}{U^{-1}})$.
Additionally, $\mu(U) = \mu(U^{-1})$ because if we let $\lambda(U) := \mu(U^{-1})$
then $\lambda$ is right-invariant, hence by Fact~\ref{f:Haar} coincides with the Haar measure $\mu$.
This completes the proof of the third inequality; the other three inequalities can be obtained similarly.
\end{proof}
For the illustration of \cref{lem:coreIneq}, see \cref{fig:coreIneqIllustration}.

\begin{figure}[t]
\center{
\includegraphics[width=0.4\linewidth]{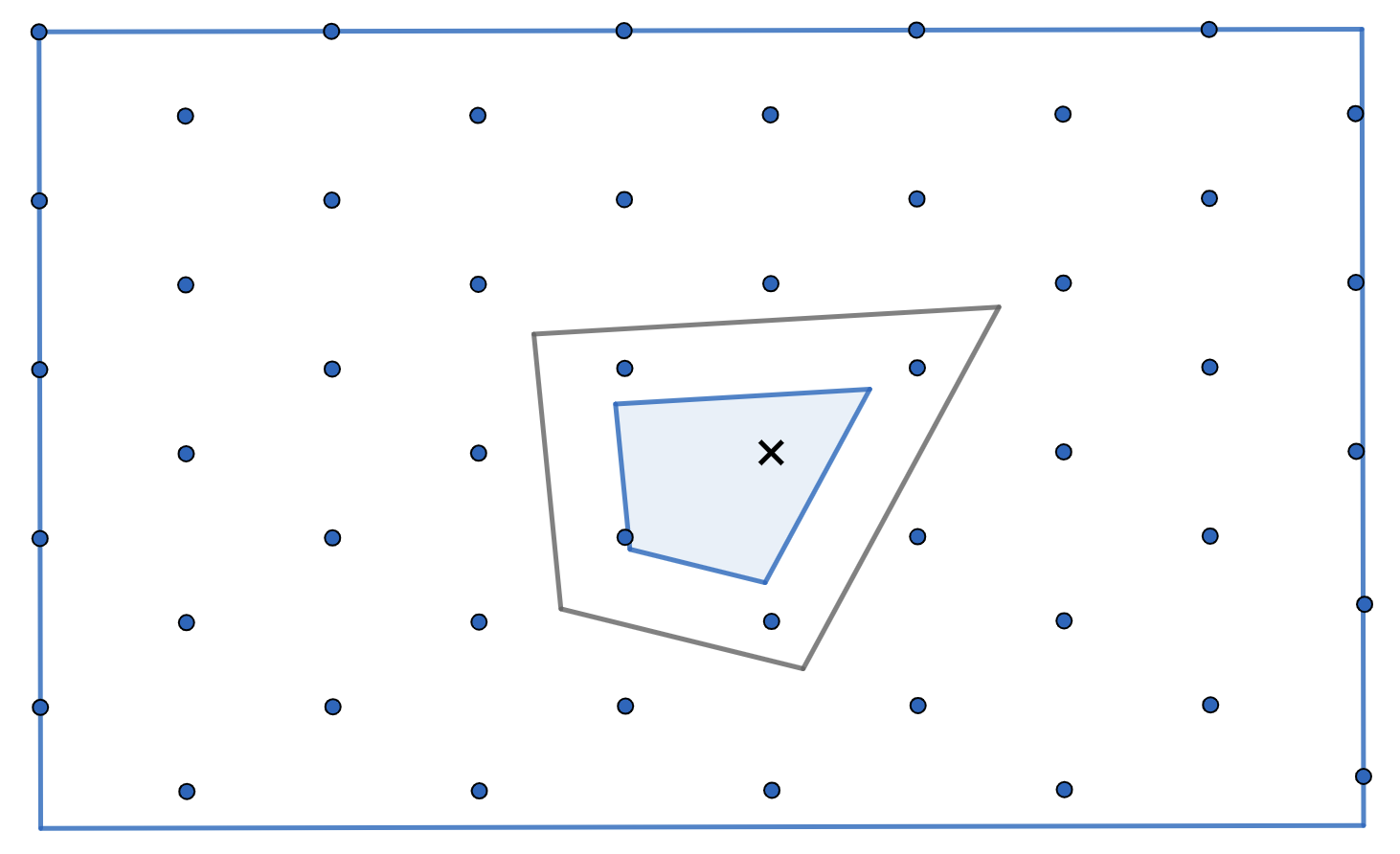}
\includegraphics[width=0.4\linewidth]{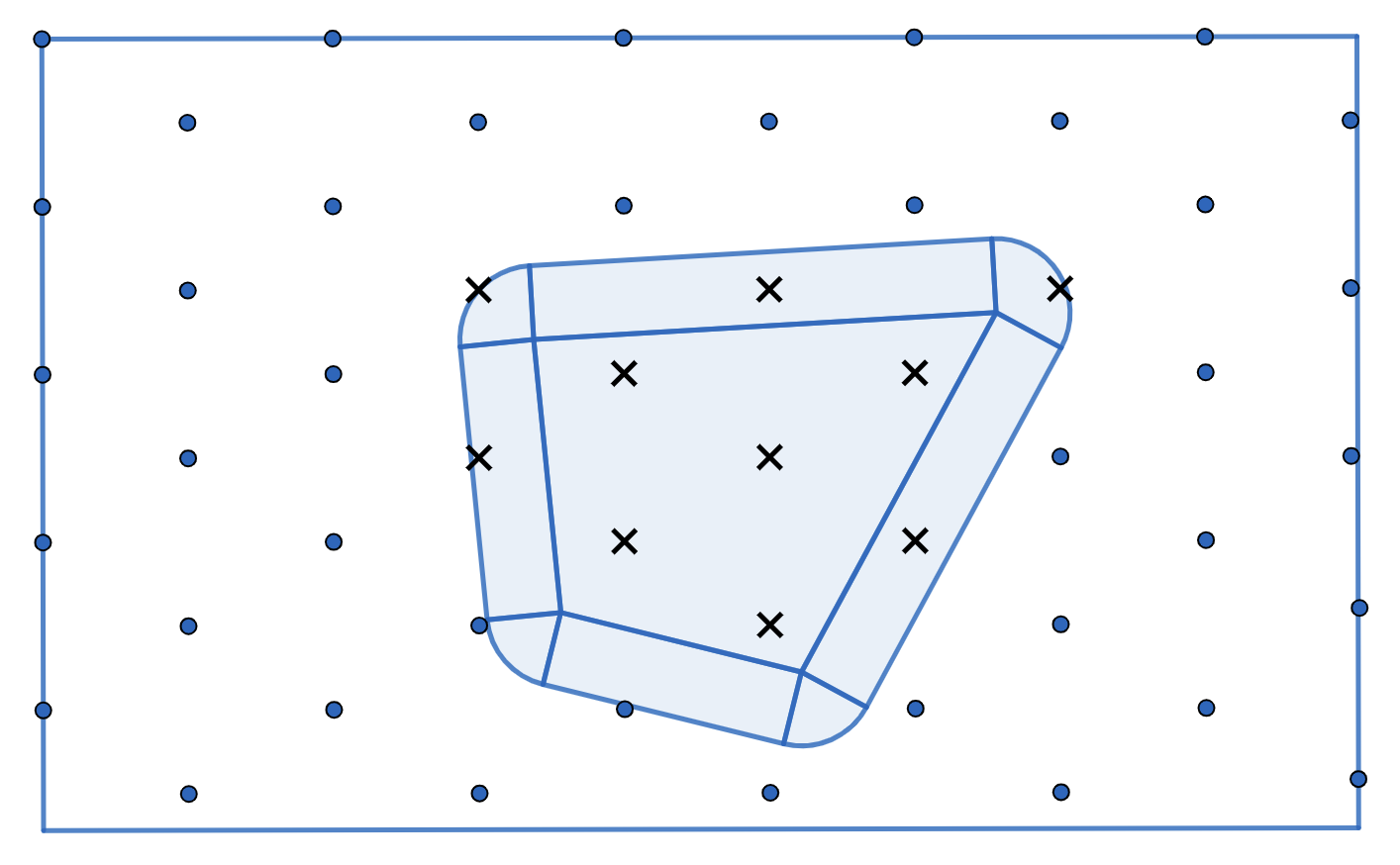}
}
\caption{Illustration of \cref{lem:coreIneq}. A blue rectangle represents the space. Blue points represent the maximum $n$-packing. A black shape represents $U$. Blue colored shapes represent inner and outer generalized balls. Counting cross-marked points and dividing it by the number of (any) points gives $\mu_{T_n}\big(\genBall{-2^{-n+1}}{U}\big)$ and $\mu_{T_n}\big(\genBall{2^{-n+1}}{U}\big)$.
}
\label{fig:coreIneqIllustration}
\end{figure}

\subsection{Algorithmic Approximation of Haar Measures}
\label{ss:Algorithm}
Our strategy to compute $\mu(U)$ with \cref{lem:coreIneq} is to compute $\mu_{T_n}$ and assert that the sequence of intervals $\big\{ \; \big[ \; \mu_{T_n}\big(\genBall{-2^{-n+1}}{U}\big), \; \mu_{T_n}\big(\genBall{2^{-n+1}}{U}\big) \; \big] \; \big\}_{n=1}^\infty$ include $\mu(U)$, and that its length converges to zero. However, there are two obstacles:
\begin{enumerate}
    \item $\mu_T$ is discrete (i.e. the value of $\mu_T(U)$ can jump even by a small perturbation to $U$), which makes it uncomputable.
    \item If $\lim_{r \rightarrow 0} \mu(\genBall{r}{U}) = \mu(U)$, then \cref{lem:coreIneq} guarantees that the length of the interval converges. However, the hypothesis may not hold in general.
\end{enumerate}
This section works around these obstacles and gives an algorithm that can compute the measure of sufficiently rich a class of sets to perform the integration.

The first thing to address is $\mu_T$. It is not computable, but procedure \method{pseudoCount} below can bound its measure on closed sets whose distances to any given points is computable. The latter condition is known as being (Turing-) located \cite{turingLocated}:

\begin{definition}
A closed subset $S$ of a computable metric space $(X,d)$ is \emph{located} if the continuous function $X\ni p \mapsto d(p,S)\in\IR$ is computable.
\end{definition}
Located sets are sometimes called computably closed sets, but being located is different from being a computable element of $\mathcal{A}(X)$.

Our workaround to the second obstacle is, instead of trying to compute the measure of every closed set, to effectively `approximate' the given set by those satisfying the convergence condition and to compute their measure. Let us define such sets first:
\begin{definition}
\label{d:regular}
On a topological space $(X, \tau)$ with a Borel measure $\mu$, call a measurable set $U$ \emph{co-inner regular} iff
\[
\mu(U) \;=\; \sup \big\{ \mu(V) \;|\; V \subseteq U \text{ open and measurable} \big\} \enspace . \]
On a compact metric group $(X, d, \circ)$ with the Haar measure $\mu$ where $d$ is bi-invariant,
a real number $r>0$ is a \emph{co-inner regular radius} iff for some/all $p\in X$, the ball $\genBall{r}{p}$ is co-inner regular.
\end{definition}
Indeed, invariance of $d$ and $\mu$ implies that $\genBall{r}{p}$ is co-inner regular iff $\genBall{r}{q}$ is.
Note that since Haar measures are regular, on a compact metric group with a bi-invariant metric and a Haar measure, if a set $U$ is co-inner regular, then $\mu(\partial U)=0$, giving $\lim_{r \rightarrow 0} \mu(\genBall{r}{U}) = \mu(U)$.

\begin{lemma}\label{lem:compMeasure}
Let $(X,d, \circ)$ be a compact topological group with bi-invariant metric $d$, Haar measure $\mu$, and computable size $\kappa_X$ of maximum packings.
    If the closure of $U$ is located and co-inner regular, then the procedure \method{computeMeasure}
     computes its measure $\mu(U)$.
\end{lemma}
Note that \method{computeMeasure} in turn calls \method{pseudoCount}$(p,S,n)$.

\begin{procedure}[t]
  \SetKwData{error}{error} \SetKwData{interval}{interval}
  \SetKwFunction{pseudoCount}{pseudoCount}
  \SetKwFunction{length}{length}
  \KwData{$\overline{U}$ located co-inner regular set, $\{T_m\}_{m=1}^\infty$ computable sequence of maximum packings, $n$ target precision}
  \KwResult{A rational number $q$ s.t. $|q - \mu(S)| \leq 2^{-n}$.}
    \caption{computeMeasure($U$, $\{T_m\}_{m=1}^\infty$, $n$)}
    \error $\leftarrow$ $\infty$\;
    $m \leftarrow 0$\;
    \While{\error $> 2^{-n}$}{
      $r \leftarrow 2^{-m}$\;
      $a \leftarrow$ \pseudoCount{$\genBall{-r}{\overline{U}}$, $T_m$, $m+1$})\;
      $b \leftarrow$ \pseudoCount{$\genBall{r/2}{\overline{U}}$, $T_m$, $m+1$})\;
      \error $\leftarrow \; b-a$ \;
      $m \leftarrow m+1$\;
      }
    \Return any $p \in$ \interval
\label{proc:computeMeasure}
\end{procedure}
\begin{procedure}[t]
\SetKwFunction{dist}{dist}
\SetKwData{count}{count}
\KwData{$S$ a located set, $T$ a finite set of points, $n$ error parameter, \dist{$p,S,m$} approximate distance between $p\in T$ and $S$ up to $2^{-m}$.}
\KwResult{A rational $q$ where $\mu_T(S) \leq q \leq \mu_T(\genBall{2^{-n}}{S})$}
\caption{pseudoCount($S$, $T$, $n$)}
$\count \leftarrow 0$\;
\ForEach{ $p \in T$}{
    \If{ \dist{$p$,$S$,$n+2$} $< 2^{-n-1}$}{$\count \leftarrow \count +1$\;}
}
\Return $\frac{\count}{|T|}$
\label{proc:pseudoCount}
\end{procedure}
\begin{proof}[Proof of \cref{lem:compMeasure}]
    In \method{computeMeasure}, $[a,b]$ is a non-empty subinterval of \linebreak $[\mu_{T_m}(\genBall{-2^{-m}}{\overline{U}}), \mu_{T_m}(\genBall{2^{-m}}{\overline{U}})]$ because
    \[ \genBall{-2^{-m}}{\overline{U}} \;\subseteq\; \genBall{2^{-(m+1)}}{\genBall{-r}{\overline{U}}} \;\subseteq\; \overline{U}
      \;\subseteq\; \genBall{2^{-(m+1)}}{\genBall{r/2}{\overline{U}}}
        \;\subseteq\; \genBall{2^{-m}}{\overline{U}} \]
    by Equation~(\ref{e:Balls}) and the postcondition of \method{pseudoCount}. This interval converges to the singleton $[\mu(U)]$ because of \cref{lem:coreIneq} and that $\lim_{|r| \rightarrow 0} \mu(\genBall{r}{\overline{U}}) = \mu(\overline{U}) = \mu(U)$ (from $\overline{U}$ being co-inner regular). The postcondition of \method{pseudoCount} is satisfied, because for any $p \in T$, every $p \in S$ is counted and every $p \notin \genBall{2^{-n}}{S}$ is not counted.
\end{proof}

Not every closed ball is co-inner regular, but `sufficiently' many are:
Co-inner regular radii can be effectively found to compute Haar measures in the form of the Haar integral by \method{findCoInnerRegularRadius}.
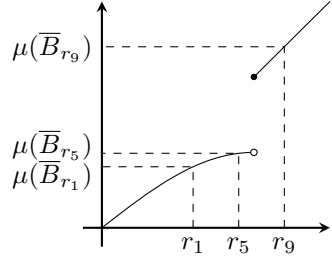
\begin{figure}[t]
    \centering

\begin{tikzpicture}[>=stealth]
    \draw[->, thick] (-0.25,0) -- (3,0);
    \draw[->, thick] (0,-0.25) -- (0,3);

    \draw (0,0) sin (2,1);
    \draw[thick] (2,1) circle (1pt);
    \fill[white] (2,1) circle (1pt);
    \filldraw[black] (2,2) circle (1pt);
    \draw (2,2) -- (3,3);

    \begin{scope}[dashed]
        \draw (1.2,0) node[anchor=north] {$r_1$} -- (1.2,0.80901699437) -- ++(-1.2,0);
        \draw (0, 0.7) node[anchor=east] {$\mu(\overline{B}_{r_1})$};
        \draw (1.8,0) node[anchor=north] {$r_5$} -- (1.8,0.98768834059) -- ++(-1.8,0);
        \draw (0,1.1) node[anchor=east] {$\mu(\overline{B}_{r_5})$};
        \draw (2.4, 0) node[anchor=north] {$r_9$} -- (2.4,2.4) -- (0,2.4) node[anchor=east] {$\mu(\overline{B}_{r_9})$};
    \end{scope}
\end{tikzpicture}

    \caption{Illustration of procedure \method{findCoInnerRegularRadius}. Here, an example graph of the discontinuous function $r \mapsto \mu(\overline{B}_r)$ is shown. Note that $\mu(\overline{B}_r(p))=\mu(\overline{B}_r(q))$ because of the invariance of the metric and the Haar measure.
    }
    \label{fig:findCoInnerRegularRadius}
\end{figure}
\cref{fig:findCoInnerRegularRadius} illustrates the procedure \method{findCoInnerRegularRadius}. $\lambda x. \text{\method{findCoInnerRegularRadius}}(a, b, \{T_m\}_{m=1}^\infty, x)$ is a nested sequence of intervals that converges to a co-inner regular radius. \method{findCoInnerRegularRadius} achieves this by recursively dividing and outputting the interval. That is, $\text{\method{findCoInnerRegularRadius}}(a, b, \{T_m\}_{m=1}^\infty, n)$ first computes the $(n-1)$-th interval ($\text{\method{findCoInnerRegularRadius}}(a, b, \{T_m\}_{m=1}^\infty, n-1)$) and outputs the $n$-th interval by dividing it.
The procedure divides the $(n-1)$-th interval into two parts $[r_1,r_5]$ and $[r_5,r_9]$, computes corresponding measures $\mu(\overline{B}_{r_1}), \mu(\overline{B}_{r_5}), \mu(\overline{B}_{r_9})$, and picks the interval which has smaller difference of measures. In this case, since $\mu(\overline{B}_{r_5})-\mu(\overline{B}_{r_1}) \leq \mu(\overline{B}_{r_9})-\mu(\overline{B}_{r_5})$, $[r_1,r_5]$ is picked. This strategy makes the difference of measures converges to zero since it is always, at least, halved on each iterations. This gives a co-inner regular radius, because in fact co-inner regular radii are continuity points of the function $r \mapsto \mu(\overline{B}_r)$.

\begin{lemma} Procedure \method{findCoInnerRegularRadius} computes a co-inner regular radius in the form of $\lambda x. \text{\method{findCoInnerRegularRadius}}(a, b, \{T_m\}_{m=1}^\infty, x)$.

\begin{procedure}[t]
    \SetKwData{interval}{interval}
    \SetKwFunction{div}{divide} \SetKwFunction{pseudoCount}{pseudoCount}
    \KwData{$a<b$ rational bounds between which to look for a co-inner regular radius,
        $\{T_m\}_{m=1}^\infty$ sequence of maximum packings, $n$ target precision.}
    \KwResult{Rational bounds $a_n, b_n$ s.t.  $(a<a_{n-1}<a_n<b_n<b_{n-1}<b) \wedge (b_n-a_n\leq 2^{-n}) \wedge |\mu(\genBall{a_n}{p}) - \mu(\genBall{b_n}{p}) | \leq 2^{-n}$
    for any/all $p \in X$.}
    \caption{findCoInnerRegularRadius($a,b, \{T_m\}_{m=1}^\infty, n$)}
    $(a_{n-1}, b_{n-1}) \leftarrow $ \findCoInnerRegularRadius{$a,b, \{T_m\}_{m=1}^\infty, n-1$}\;
    $r_1, r_5, r_9 \leftarrow \frac{9a_{n-1} + 1b_{n-1}}{10}, \frac{5a_{n-1} + 5b_{n-1}}{10}, \frac{1a_{n-1} + 9b_{n-1}}{10}$\;
    Pick sufficiently large $N$ s.t. $2^{-N+2} \leq \frac{b_{n-1}-a_{n-1}}{10}$\;
    Compute an element $p \in X$ using the fact that $X$ is a computable metric space\;
    $m_1, m_5, m_9 \leftarrow$ \pseudoCount{$\genBall{r_1}{p}, T_N, N$}, \pseudoCount{$\genBall{r_5}{p}, T_N, N$}, \pseudoCount{$\genBall{r_9}{p}, T_N, N$}\;
    $\epsilon \leftarrow \frac{b_{n-1}-a_{n-1}}{10}$\;
    \leIf{$m_9 - m_5 \leq m_5 - m_1$}{\Return $[r_1+\epsilon , r_5-\epsilon ]$}{\Return $[r_5+\epsilon , r_9-\epsilon ]$}
    \label{proc:findCoInnerRegularRadius}
\end{procedure}

\begin{proof}
    $\lambda x. \text{\method{findCoInnerRegularRadius}}(a, b, \{T_m\}_{m=1}^\infty, x)$ represents $r := \lim_{n \rightarrow \infty} a_n$, where $a_n$ is the first element of the interval that \method{findCoInnerRegularRadius} outputs. $r$ is a co-inner regular radius because the fact $r \in (a_n, b_n)$ makes $\partial \genBall{r}{p} \subseteq \genBall{b_n}{p} \setminus \genBall{a_n}{p}$, which leads to $\mu(\partial \genBall{r}{p}) \leq |\mu(\genBall{b_n}{p}) - \mu(\genBall{a_n}{p})| \leq 2^{-n}$ for any $n$. This implies $\mu(\partial \genBall{r}{p}) = 0$.

    Now it is sufficient to prove the postconditions. Let us only prove $\mu(\genBall{a_n}{p})-\mu(\genBall{b_n}{p})$ $\leq 2^{-n}$, since others are straightforward. Let $r_i := \frac{ia_{n-1}+(10-i)b_{n-1}}{10}$ and $m_i := \mu(\genBall{r_i}{p})$. Because of \cref{lem:coreIneq} and the fact that $N$ is sufficiently large, $\mu(\genBall{r_i}{p}) \leq \mu_{T_N}(\genBall{r_{i+1}}{p}) \leq m_{i+1} \leq \mu_{T_N}(\genBall{r_{i+1}+2^{-N}}{p}) \leq \mu(\genBall{r_{i+2}}{p})$. Then since $2^{-n+1} \geq |\mu(\genBall{a_{n-1}}{p}) - \mu(\genBall{b_{n-1}}{p})|$ $\geq |m_9 - m_1| \geq |m_9 - m_5| + |m_5 - m_1|$, WLOG $|m_5 - m_1| \leq 2^{-n}$. Then $|\mu(\genBall{r_5-\epsilon}{p}) - \mu(\genBall{r_1+\epsilon}{p})| \leq |m_5 - m_1| \leq 2^{-n}$.
\end{proof}
\end{lemma}

\subsection{Main Algorithm for Haar Integration}
\begin{proof}[Explicit algorithm of \cref{theo:explicit}]

The procedure \method{computeIntegral} computes the Haar integral $\int_X f\,d\mu$.
Generalizing classical Riemann sums,
it partitions $X$ into subsets $U_i$, $i\leq N$, of sufficiently small diameter (see \cref{fig:partition}):
given by a modulus of continuity
such that $f$ on each $U_i$ varies by at most $2^{-n}$.
Then it sums those values of $f|_{U_i}$,
each weighted by the measure of $U_i$.
In order to invoke \method{computeMeasure},
we want the $U_i$ to be located and co-inner regular:
as provided by \method{findNicePartition}.
Specifically, each $U_i$ will be of the form
$U_i=\genBall{R}{p_i} \setminus \bigcup_{j <i} \genBall{R}{p_j}$
for $p_1,\ldots,p_N\in T_m$ and real $R>0$ provided by \method{findCoInnerRegularRadius}.

\begin{figure}[t]
    \centering
    \includegraphics[width=0.5\linewidth]{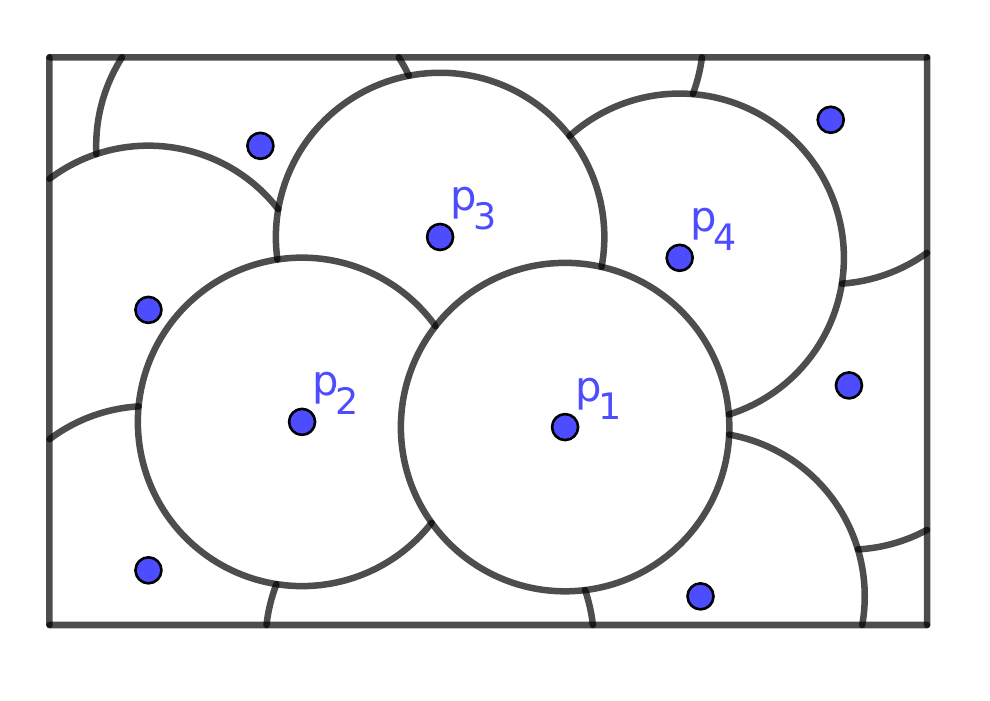}
    \caption[How the space is partitioned]{Consider the whole rectangle as the whole space. Then this is how the procedure \method{computeIntegral} partitions the space. For example, the subset containing $p_1$ is $\genBall{R}{p_1}$. Similarly, the subset containing $p_2$ is $\genBall{R}{p_2} \setminus \genBall{R}{p_1}$, the subset containing $p_3$ is $\genBall{R}{p_3} \setminus (\genBall{R}{p_1} \cup \genBall{R}{p_2})$, and so on. Then the subsets of the partition are of the form $\genBall{R}{p_i} \setminus \bigcup_{j <i} \genBall{R}{p_j}$.}
    \label{fig:partition}
\end{figure}

\begin{procedure}[t]
  \SetKwFunction{findNicePartition}{findNicePartition}\SetKwFunction{computeMeasure}{computeMeasure} \SetKwFunction{center}{center} \SetKwFunction{modulus}{modulus}
  \SetKwFunction{bound}{bound}
  \KwData{real function $f$, sequence of maximum packings $\{T_m\}_{m=1}^\infty$, target precision $n$}
  \KwResult{a rational number $q$ s.t. $|q - \int_X f d\mu| \leq 2^{-n}$}
    \caption{computeIntegral($f$, $\{T_m\}_{m=1}^\infty$, $n$)}
    $m_f$ $\leftarrow$ \modulus{$f$, $n+1$} \tcp*[r]{\modulus is from \cite[Definition~2.12]{Ko91}}
    $\{U_i\}_{i=1}^N \leftarrow$ \findNicePartition{$\{T_m\}_{m=1}^\infty$, $m_f$}\;
    $M \leftarrow$ \bound{$|f|$}\;
    \ForEach{$U_i$ in $\{U_i\}_{i=1}^N$}{
    $p_i \leftarrow$ \center{$U_i$}\;
    $m_i \leftarrow $\computeMeasure{$U_i$, $\{T_m\}_{m=1}^\infty$, $n+1+i+\log{M}$}\;
    }
    \Return $\sum_{i=1}^N m_if(p_i)$
    \label{proc:computeIntegral}
\end{procedure}
\begin{procedure}[t]
    \SetKwFunction{findCoInnerRegularRadius}{findCoInnerRegularRadius}
    \KwData{$\{T_m\}_{m=1}^\infty$ is a sequence of maximum packings, $n$ is the target precision}
    \KwResult{A partition $P = \{U_i\}_{i=1}^N$ s.t. each $\overline{U_i}$ is a located co-inner regular set of the form $U_i=\genBall{R}{p_i} \setminus \bigcup_{j<i} \genBall{R}{p_j}$.}
    $P \leftarrow \{\}$\;
    $R \leftarrow \lambda x.$ \findCoInnerRegularRadius{$(2^{-n-1}, 2^{-n}), \{T_m\}_{m=1}^\infty, x$}\;
    \ForEach{$p_i$ in $T_{n+1}$}{
        $U_i \leftarrow {\genBall{R}{p_i} \setminus \bigcup_{U \in P} U}$\;
        $P \leftarrow P \cup \{U_i\}$\;
    }
    \Return $P$
    \caption{findNicePartition($\{T_m\}_{m=1}^\infty$, $n$)}
    \label{proc:findNicePartition}
\end{procedure}
\end{proof}

\begin{lemma}
    Procedure \method{computeIntegral} is correct
\end{lemma}
\begin{proof}
    $m_f$ is a modulus of continuity \cite[Definition~2.12]{Ko91} of $f$ with precision $n+1$. This means $d(x,y) \leq 2^{-m_f} \Rightarrow |f(x)-f(y)| \leq 2^{-n-1}$. The partition $P = \{U_i\}_{i=1}^N$ satisfies that every $U_i$ has radius smaller than $2^{-m_f}$. So
    \begin{equation*}
        \begin{split}
        |q - \int_X fd\mu | &= |\sum_{i=1}^N m_if(p_i) - \int_X fd\mu | \\
        &\leq \sum_{i=1}^N 2^{-n-1-i-\log{M}}f(p_i) + |\sum_{i=1}^N \mu(U_i)f(p_i) - \int_X fd\mu | \\
        &\leq \sum_{i=1}^N 2^{-n-1-i-\log{M}}M + \sum_{i=1}^N \mu(U_i)2^{-n-1} \\
        &\leq 2^{-n-1}2^{-\log{M}}M + 2^{-n-1}\sum_{i=1}^N \mu(U_i) \leq 2^{-n}
        \end{split}
    \end{equation*}
\end{proof}

\begin{lemma}
    Procedure \method{findNicePartition} is correct.
\end{lemma}
\begin{proof}The followings are proofs of postconditions in the same order in the pseudocode.
    \begin{enumerate}
        \item $\{U_i\}_{i=1}^N$ is a partition because they are disjoint, and they covers the whole space since $X = \bigcup_{p \in T_{n+1}} \genBall{2^{-n-1}}{p} \subseteq \bigcup_{U \in P} U$.
        \item $\overline{U_i}$ are clearly closed. They are located because $\bigcup_{U \in P} U = \bigcup \genBall{R}{p_i} \Rightarrow \overline{U_i}$ is of the form $\overline{\genBall{R}{p} \setminus \bigcup \genBall{R}{p_i}}$ and every $R, p_i$ are computable.
        \item $\overline{U_i}$ are co-inner regular because $\genBall{R}{p_i}$ are co-inner regular and it is preserved under intersection, union, complement, and closure. $\genBall{R}{p_i}$ are co-inner regular because the postcondition of the \method{findCoInnerRegularRadius} ensures that $R$ is a co-inner regular radius.
        \item $U_i \subseteq \genBall{R}{p_i} \subseteq \genBall{2^{-n}}{p_i} \Rightarrow U_i$ is contained in a closed ball with radius $2^{-n}$.
    \end{enumerate}
\end{proof}

\section{Discussion of Hypotheses}\label{sec:compare}
While the requirements of Theorem \ref{theo:synthetic:nice} and Theorem \ref{theo:explicit} appear to be very different, it turns out that actually, both theorems are applicable in the very same cases.

For one direction, suppose we have a computably compact computable metric space $(\mathbf{X},d)$ with a computable group operation $\circ$.
Then\footnote{We appreciate relevant discussion on MathOverflow \cite{mo-biinvariant}.}
\[d'(a,b)  \;:=\; \sup\nolimits_{x \in \mathbf{X}} \sup\nolimits_{y \in \mathbf{X}} d(x \circ a \circ y, x \circ b \circ y)\]
constitutes a topologically equivalent and also computable, but now bi-invariant, metric.
The size of maximum packings may be non-computable for a CCCMS:

\begin{example}
There exists a CCCMS $\mathbf{X}$ such that $\kappa$ defined as $$\kappa(n) = \max \{|T| \mid T \subseteq \mathbf{X} \ \forall x \neq y \in T \ \ d(x,y) > 2^{-n}\}$$ is Turing-equivalent to the Halting problem.
\begin{proof}
Let $s_n = 1 + 2^{-5-t_n}$ if the $n$-th Turing machine halts after exactly $t_n$ steps, and $s_n = 1$ if the $n$-th Turing machine never halts. Clearly, $(s_n)_{n \in \mathbb{N}}$ is a computable sequence. We now modify the standard metric on Cantor space to be $d(p,q) = s_n2^{-n}$ where $n$ is the first component where $p$ and $q$ differ. We find $(\Cantor,d)$ to be CCCMS (in fact, to be computably isomorphic to $\Cantor$), but if $\kappa$ codes the size of maximum packings in $(\Cantor,d)$, then $\kappa(n) = 2^n$ iff the $n$-th Turing machine halts, and $\kappa(n) = 2^{n-1}$ otherwise.
\end{proof}
\end{example}

However, for any CCCMS there is a computable sequence of radii converging to zero for which we can compute the maximum packings (see Subsection \ref{subsec@computingmaximumpackings} below). It is straightforward to see that this suffices for Theorem \ref{theo:explicit}. As such, we see that the requirements for Theorem \ref{theo:explicit} are implied by those of Theorem \ref{theo:synthetic:nice}.

For the converse direction, note that Theorem~\ref{theo:explicit} does not
suppose the group operation $\circ$ (nor neutral element nor inversion) to be computable.
Indeed, a group operation on a CCCMS can have a computable bi-invariant metric but fail to be computable itself.
This is due to the potential for many different group operations to have the same bi-invariant metric:

\begin{example}
Fix some $A \subseteq \mathbb{N}$. Let $G_n := IZ_{p_n^2}$ if $n \in A$, and let $G_n := \IZ_{p_n} \times \IZ_{p_n}$ if $n \notin A$, where $p_n$ is the $n$-th prime.
Note that both have cardinality $p_n^2$ but are not isomorphic as additive groups yet both have the same Haar measure under the bi-invariant discrete metric.
Now let $G_A := \Pi_{n \in \mathbb{N}} G_n$, equipped with the Baire space metric.
For $A \neq B$ we find that $G_A$ and $G_B$ are not homeomorphic. The group operation on $G_A$ is computable iff $A$ is decidable. However, both the bi-invariant metric structure on $G_A$ and the Haar measure are all independent of $A$, and computable.
\end{example}
Interestingly, the Haar measure on a compact group is determined already by an invariant metric
and independent of the potentially many different underlying group operations:

\begin{corollary}
\label{corr:uniqueness}
Consider a compact metric space $(X,d)$ with two group operations $\circ$ and $\circ'$
both rendering $d$ left-invariant. Then $(X,\circ)$ and $(X,\circ')$ induce the same Haar measure.
\end{corollary}
\begin{proof}
In the metric case,
the net of neighborhoods $B$ of $e$ from Fact~\ref{f:Classical}
becomes a sequence of open balls $B_{1/2^n}(e)$.
Left-invariance implies that all translates
$q\cdot B_{r}(e)=B_{r}(q)$ have the same measure;
and by the group property,
every open ball $B_{r}(q)$ is a translate of $B_r(e)$:
for both $\circ$ and $\circ'$.
In particular, $[X:B_{1/2^n}]=\kappa_X(n)$.
\end{proof}
On the other hand the collection of different group operations
$\circ$ to a given bi-invariant metric $d$ is `tame':

\begin{lemma}
\label{lemma:groupoperationscompact}
Let $(\mathbf{X},d)$ be a CCCMS. The set $O \subseteq \mathcal{C}(\mathbf{X} \times \mathbf{X},\mathbf{X})$ of group operations
rendering $d$ bi-invariant is a computably compact set.
\begin{proof}
If $d$ is bi-invariant for $\circ \in \mathcal{C}(\mathbf{X} \times \mathbf{X},\mathbf{X})$,
the triangle inequality gives
\[ d(a\circ x,b\circ y) \;\leq\; d(a\circ x,b\circ x)\;+\;d(b\circ x,b\circ y)
\;=\; d(a,b)+d(x,y) \leq 2 d((a,x),(b,y)) \quad ; \]
rendering $\circ$ 2-Lipschitz with respect to the maximum metric on $\mathbf{X}\times\mathbf{X}$.
By the effective Arzel\`a-Ascoli theorem,
the subset of 2-Lipschitz $f\in\mathcal{C}(\mathbf{X} \times \mathbf{X},\mathbf{X})$ is computably compact.
Within this set, we are interested in those satisfying the bi-invariance and group axioms:
\begin{multline*}
\forall a,b,c: \qquad  d\big(f(c,a), f(c,b)\big) \;=\; d(a, b) \;=\; d\big(f(a,c),f(b,c)\big)
\\  \exists e,a':
\quad  f(a,e) = a = f(e,a) \wedge f(f(a,b),c) = f(a,f(b,c)) \wedge f(a,a') = e = f(a',a)
\end{multline*}
These are computably closed predicates since the quantification is over the computably compact and computably overt space $\mathbf{X}$. This ends the proof, since a computably closed subset of a computably compact set is computably compact.
\end{proof}
\end{lemma}
We can combine Corollary \ref{corr:uniqueness} and Lemma \ref{lemma:groupoperationscompact} to see that Theorem \ref{theo:synthetic:main} also implies that from a CCCMS $(X,d)$ such that some group operation is bi-invariant for $d$ we can compute the Haar measure for any such group operation.

\subsection{Computing maximum packings}
\label{subsec@computingmaximumpackings}

Let $\mathbb{R}^+$ denote the positive real numbers, and $\mathbb{Q}^+$ the positive rational numbers. We shall set out to prove:

\begin{theorem}
\label{theo:packings}
The following are equivalent for a computable Polish space $\mathbf{X}$:
\begin{enumerate}
\item $\mathbf{X}$ is computably compact.
\item There are computable sequences $(\varepsilon_i \in \mathbb{Q}^+)_{i \in \mathbb{N}}$ and $(s_i \in \mathbb{N})_{i \in \mathbb{N}}$ such that $\lim_{i \to \infty} \varepsilon_i = 0$ and \[s_i = \sup \{m \in \mathbb{N} \mid \exists x_1,\ldots,x_m \in \mathbf{X} \ \forall i, j \leq m, i \neq j \ d(x_i,x_j) > \varepsilon_i\}\]
\end{enumerate}
\end{theorem}

Our proof involves a number of lemmata, which we shall gather first.

\begin{lemma}
\label{lemma:boundsfunctions}
Let $\ell : \mathbb{R}^+ \to \mathbb{N}_<$ and $u : \mathbb{R}^+ \to \mathbb{N}_>$ be computable functions with $\forall x \in \mathbb{R}^+ \ \ell(x) \leq u(x)$ and $\forall x, \varepsilon \in \mathbb{R}^+ \ u(x) \leq \ell(x+\varepsilon)$. Then $U_n = \{x \in \mathbb{R}^+ \mid \ell(x) = u(x) = n\}$ defines a computable sequence of connected open sets satisfying that $\mathbb{R}^+ \setminus \left (\bigcup_{n \in \mathbb{N}} U_n \right )$ consists of countably many isolated points, and that $\forall n, k \in \mathbb{N} \ x \in U_n \wedge y \in U_{n+k+1} \Rightarrow x < y$.
\begin{proof}
Since $\mathalpha{\leq} \subseteq \mathbb{N}_> \times \mathbb{N}_<$ is computably open, we can obtain any $\{x \in \mathbb{R}^+ \mid \ell(x) \geq u(x) = n\}$ as an open set, and by the assumed properties of $\ell$ and $u$, this is actually equal to $U_n$. Putting together the two inequalities for $u$ and $\ell$ shows that both are non-decreasing functions, which implies the linear order on the $U_n$.

Our arguments already established that $\mathbb{R}^+ \setminus \left (\bigcup_{n \in \mathbb{N}} U_n \right )$ is a union of sets of the form $\{x \mid \forall y \in U_n \ \forall z \in U_m \ y < x < z\}$, where $m$ is the least number larger than $n$ such that $U_m \neq \emptyset$. Assume that $\{x \mid \forall y \in U_n \ \forall z \in U_m \ y < x < z\}$ is not a singleton. Then there there are $x_0 < x_1 < x_2 \ldots$ contained in it (as we are dealing with an interval). We know that $n \leq \ell(x_k) < u(x_k) \leq \ell(x_{k+1}) \leq m$ for all $k \in \mathbb{N}$, so there would need to be an infinite strictly increasing sequence between $n$ and $m$ in $\mathbb{N}_<$, contradiction.
\end{proof}
\end{lemma}

\begin{corollary}
\label{corr:separations}
Given $\ell, u$ as in Lemma \ref{lemma:boundsfunctions}, we can compute a pair of sequences $(a_i \in \mathbb{Q}^+)_{i \in \mathbb{N}}$ and monotone $(k_i \in \mathbb{N})_{i \in \mathbb{N}}$ such that $(k_i)_{i \in \mathbb{N}}$ is cofinal in the range of $\ell$; and $\ell(a_i) = u(a_i) = k_i$ for all $i \in \mathbb{N}$.
\end{corollary}

\begin{lemma}
\label{lemma:ell}
Given a compact computable metric space $\mathbf{X}$, the map $\ell : \mathbb{R}^+ \to \mathbb{N}_<$ defined via \[\ell(\delta) = \sup \{m \in \mathbb{N} \mid \exists x_1,\ldots,x_m \in \mathbf{X} \ \forall i, j \leq m \ d(x_i,x_j) > \delta^{-1}\}\]
is computable.
\begin{proof}
$\forall i, j \leq m \ d(x_i,x_j) < \delta^{-1}$ defines an open predicate on $\mathbf{X}^m \times \mathbb{R}^+$, and every computable metric space is computably overt. Compactness of $\mathbf{X}$ ensures that the function is well-defined.
\end{proof}
\end{lemma}

\begin{lemma}
\label{lemma:u}
Given a computably compact computable metric space $\mathbf{X}$, the map $u : \mathbb{R}^+ \to \mathbb{N}_>$ defined via \[u(\delta) = \inf \{m \in \mathbb{N} \mid \forall x_1,\ldots,x_{m+1} \in \mathbf{X} \ \exists  i, j \leq m+1 \ i \neq j \wedge d(x_i,x_j) < \delta^{-1}\}\]
is computable.
\begin{proof}
$\exists  i, j \leq m+1 \ i \neq j \wedge d(x_i,x_j) < \delta^{-1}$ defines an open predicate on $\mathbf{X}^{m+1} \times \mathbb{R}^+$, and computable compactness lets us do the universal quantification over $\mathbf{X}^{m+1}$.
\end{proof}
\end{lemma}

\begin{lemma}
\label{lemma:applicable}
The functions $\ell$ and $u$ defined in Lemmas \ref{lemma:ell} and \ref{lemma:u} satisfy the criteria of Lemma \ref{lemma:boundsfunctions}.
\begin{proof}
Computability of the maps is established in the lemmas defining them. From the definitions, it is clear that $\ell(\delta) \leq u(\delta)$. Now assume that $u(\delta) = m$ holds for some $\delta \in \mathbb{R}^+$, $m \in \mathbb{N}$. Since $u(\delta)$ is defined as an infimum, this means that there are $x_1,\ldots,x_m \in \mathbf{X}$ such that $d(x_i,x_j) \geq \delta^{-1}$ whenever $i \neq j$. But then for any $\varepsilon \in \mathbb{R}^+$, we also have $d(x_i,x_j) > (\delta + \varepsilon)^{-1}$, and hence the $x_i$ also witness that $\ell(x + \varepsilon) \geq m$, establishing the second inequality.
\end{proof}
\end{lemma}

\begin{lemma}
\label{lemma:scale}
Let $\mathbf{X}$ be a compact metric space with $\mathbf{X} = \bigcup_{n = 0}^N B(z_n,r_n)$. There is some $e > 0$ such that for $e > \varepsilon > 0$ such that for any $x \in \mathbf{X}$ there exists some $n \leq N$ with $\overline{B}(x,\varepsilon) \subseteq B(z_n,r_n)$.
\begin{proof}
The map $x \mapsto \sup \{\delta \in \mathbb{R}^+ \mid \exists n \leq N \ d(x,z_n) + \delta^{-1} < r_n\}$ is a continuous function from $\mathbf{X}$ to $\mathbb{R}_<$, and as such has some upper bound $e^{-1}$.
\end{proof}
\end{lemma}

\begin{proof}[Proof of Theorem \ref{theo:packings}]
The implication from $(1)$ to $(2)$ proceeds via Corollary \ref{corr:separations}, which is applicable by Lemma \ref{lemma:applicable}.

To see that $(2)$ implies $(1)$, assume that we are given some $U \in \mathcal{O}(\mathbf{X})$ in the form $U = \bigcup_{n \in \mathbb{N}} B(z_n,r_n)$. For each $i \in \mathbb{N}$, we try to find $x_1,\ldots,x_{s_i} \in U$ such that $\forall j,k \leq s_i, k \neq j \ d(x_k,x_j) > \varepsilon_i$ (which will succeed if indeed $U = \mathbf{X}$. If in addition, we find that for each $j \leq m$ there is some $k_j \in \mathbb{N}$ with $d(x_j,z_k) + \varepsilon_i < r_k$ (which we can indeed detect if true), then we answer that $U = \mathbf{X}$. We are justified in doing so, because we already know that $\mathbf{X} = \bigcup_{j \leq s_i} \overline{B}(x_j,\varepsilon_j)$.

It remains to argue that we correctly identify $U = \mathbf{X}$ in all cases. The existence of the sequences alone, together with completeness of $\mathbf{X}$, establishes that $\mathbf{X}$ is compact. Thus, if $U = \mathbf{X}$, then already $\mathbf{X} = \bigcup_{n = 0}^N B(z_n,r_n)$ for some $N \in \mathbb{N}$. Choosing $i$ sufficiently large to make sure that $\varepsilon_i$ is sufficiently small for Lemma \ref{lemma:scale} to apply ensures that the procedure above will correctly give a positive answer.
\end{proof}

\subsection{More on computable compactness}

To conclude this section, we shall consider a family of examples that show that we need more computability requirements than that of the metric and of the group operation. We consider the closed subgroups of $(\Cantor,\oplus)$, where $\oplus$ denotes the componentwise exclusive or. These subgroups are of the form $$G_A := \{p \in \Cantor \mid \forall n \in A \ p(n) = 0\}$$ for some $A \subseteq \mathbb{N}$. Each $G_A$ inherits compactness, computable metrizability and the computability of the group operation from $(\Cantor, \oplus)$.

$G_A$ is computably compact iff $A$ is c.e., and effectively separable (and thus a computable metric space) iff $A$ is co-c.e. Now if we have the Haar measure $\mu_A$ on $G_A$, we can recover $A$ since $\mu_A(\{p \in G_A \mid p(n) = 1\}) = \frac{1}{2}$ iff $n \notin A$ and $\mu_A(\{p \in G_A \mid p(n) = 1\}) = 0$ iff $n \in A$. We thus see that $G_A$ is a CCCMS iff $\mu_A$ is computable -- so neither computable compactness or computable separability are dispensable for the computability of the Haar measure.

If we already have a bi-invariant metric, computable compactness is even necessary:
\begin{theorem}
\label{theo:invariantimpliescompact}
Let $(\mathbf{X},d)$ be a computable metric space with computable probability measure $\mu$
such that $\mu(B_r(x))$ depends only on $r$ but not on $x$. Then $\mathbf{X}$ is computably compact.
\end{theorem}

\begin{theorem}
\label{theo:epsiloncoverings}
The following are equivalent for a complete computable metric space $(\mathbf{X},d)$:
\begin{enumerate}
\item $\mathbf{X}$ is computably compact.
\item There is a computable multi-valued function that on input $\varepsilon > 0$ computes a tuple $(x_1,\ldots,x_N) \in \mathbf{X}^N$ such that $\mathbf{X} \subseteq \bigcup_{n = 1}^N B(x_n,\varepsilon)$.
\end{enumerate}
\begin{proof}
The implication from $1.$ to $2.$ is straight-forward: By computable compactness we can recognize a suitable solution, and compactness ensures that there is one.

For the converse, we first note that the well-definedness of the multi-valued function in $2.$ is just total boundedness, which for a complete metric space implies compactness. To see that we even get computable compactness, we assume that we are given as input some $U \in \mathcal{O}(\mathbf{X})$, which we can take to be of the form $U = \bigcup_{n \in \mathbb{N}} B(z_n,r_n)$. We now test for each $k \in \mathbb{N}$ the following:

We obtain some $(x_1,\ldots,x_M)$ such that $\mathbf{X} \subseteq \bigcup_{m = 1}^M B(x_m,2^{-k})$. If for every $m \leq M$ there exists some $n \in \mathbb{N}$ with $d(z_n,x_m) + 2^{-k} < r_n$, then clearly $B(x_m,2^{-k}) \subseteq B(z_n,r_n)$, and hence $\mathbf{X} \subseteq U$. Lemma \ref{lemma:scale}
guarantees that if indeed $\mathbf{X} = U$, then for sufficiently large $k$ this test will be successful.
\end{proof}
\end{theorem}

\begin{proof}[of Theorem \ref{theo:invariantimpliescompact}]
First, we can conclude that $\mathbf{X}$ is compact: If not, there would be some $2\varepsilon$-separated sequence $(a_n)_{n \in \mathbb{N}}$, but then $\sum_{n \in \mathbb{N}} \mu(B(a_n,\varepsilon))$ is a constant series and diverges, but also would need to be bounded from above by $\mu(\mathbf{X})$.

To see that $\mathbf{X}$ is even computably compact, we use Theorem \ref{theo:epsiloncoverings} and prove instead the computability of finite $\varepsilon$-coverings. To obtain a $\varepsilon$-covering, we try out finite unions of balls of radius $0.5\varepsilon$, and compute their measure. As $\mathbf{X}$ is compact, some $0.5\varepsilon$-covering exists. In particular, we will eventually find a finite union $T_k$ whose measure is at least $1 - 2^{-k}$ for any $k$. We simultaneously compute lower bounds for the measure of $B(x,0.4\varepsilon)$, until we find one exceeding $2^{-k}$. Now we claim that the $\varepsilon$-balls around the centers from $T_k$ do cover $\mathbf{X}$: Assume that they do not. Then some $y \in \mathbf{X}$ is not covered. But then $B(y,0.4\varepsilon)$ is disjoint from the union of the $0.5\varepsilon$-balls around centers from $T_k$. But this is contradicted by the sum of the measure of these sets being greater than $1$.
\end{proof}

\section{Computational Complexity of the Haar Integral}
\label{sec:Complexity}

We now move beyond mere computability of the Haar measure, and consider the computational complexity of this task for the groups $G=\calSO(3)$, $G=\calO(3)$, $G=\calSU(2)$, and $G=\calU(2)$. In each case, the complexity turns out to be closely related to the complexity class $\classSharpP_1$. We prove Theorem~\ref{theo:Complexity}, namely

\begin{enumerate}
\item[a)]
For every polynomial-time computable $f\in\calC(G)$,
$\int_G f\in\IR$ is computable in polynomial space (and exponential time).
\item[b)]
If $\classFP_1=\classSharpP_1$ and
$f\in\calC(G)$ is polynomial-time computable,
then so is $\int_G f\in\IR$.
\item[c)]
There exists a polynomial-time computable $f\in\calC(G)$
such that polynomial-time computability of $\int_G f\in\IR$
implies $\classFP_1=\classSharpP_1$.
\end{enumerate}

To this end recall \cite[Theorem~5.32]{Ko91} that
(a), (b), and (c) are known for definite Riemann integration
\[ \calC[0;1]\;\ni\; \tilde f \;\mapsto\; \int\nolimits_0^1 \tilde f(t)\,dt \;\in\; \IR  \enspace . \]
Moreover, Item~(c) remains true for $\tilde f\in\calC_0^\infty[0;1]$:
the class of smooth (infinitely often differentiable) $\tilde f:[0;1]\to\IR$ such that $\tilde f(0)=0=\tilde f(1)$; cmp. \cite{DBLP:conf/macis/FereeZ15,lmcs:3924}.

Before proceeding to the groups $\calSO(3)$, $\calO(3)$, $\calSU(2)$, $\calU(2)$,
recall the argument for the case
$\calU(1)=\big\{ \exp(2\pi i t) : 0\leq 1\leq 1\big\}$ equipped with complex multiplication and the Haar integral
\[ \calC\big(\calU(1)\big)\;\ni\; f\;\mapsto\;\int\nolimits_0^1 f\big(\exp(2\pi i t)\big)\,dt \enspace :\]
So to see (c), consider the polynomial-time computable embedding
\[ \calC_0[0;1] \;\ni\; \tilde f \;\mapsto\; \big( \exp(2\pi i t) \mapsto \tilde f(t) \big) \;\in\; \calC\big(\calU(1)\big) \enspace . \]
And to see (a) and (b) for $G=\calU(1)$, consider the polynomial-time computable embedding
\[ \calC\big(\calU(1)\big) \;\ni\; f \;\mapsto\; \big( t \mapsto f(\exp(2\pi i t)) \big) \;\in\; \calC[0;1] \enspace . \]
This also covers $\calSO(2)\cong\calU(1)$; and integration over $\calO(2)\cong\calSO(2)\times\{\pm1\}$
amounts to two integrals over $\calSO(2)$.

\medskip
Let $\IH=\{ \alpha+i\beta+j\gamma+k\delta : \alpha,\beta,\gamma,\delta\in\IR\}$
denote the quaternions, parameterized as real quadruples with respect to units $1,i,j,k$.
The group $\calSU(2)$ is well-known, and easily verified to be, 
isomorphic to the multiplicative group $\IH_1$ of quaternions of norm 1 (aka \emph{versors}) via isomorphism
\begin{equation}
\label{e:Quaternion}
\IH_1 \;\ni\; \alpha+i\beta+j\gamma+k\delta \;\mapsto\; \left( \begin{array}{cc}
\alpha+i\beta & -\gamma+i\delta \\
\gamma+i\delta & \alpha-i\beta \end{array}\right) \;\in\; \calSU(2)
\end{equation}
with $|\alpha|^2+|\beta|^2+|\gamma|^2+|\delta|^2=1$.
Reparameterize $\IH_1$ in generalized spherical coordinates
\begin{multline*}
[0;\pi)\times[0;\pi)\times[0;2\pi) \;\ni\; (\eta,\vartheta,\varphi) \;\mapsto\; \Psi(\eta,\vartheta,\varphi) \;:=\; \\
\cos(\eta) + i\sin(\eta)\cos(\vartheta)
+j\sin(\eta)\sin(\vartheta)\cos(\varphi)+k\sin(\eta)\sin(\vartheta)\sin(\varphi) \;\in\; \IH_1
\end{multline*}
with Jacobian determinant
$\big|\det\big(\Psi'(\eta,\vartheta,\varphi)\big)\big|=\sin^2(\eta)\sin(\vartheta)$,
and verify that integration by change-of-variables
\begin{equation}
\label{e:Integral}
\calC(\IH_1) \;\ni\; f \;\mapsto \; \int\nolimits_0^{2\pi} \int\nolimits_0^\pi \int\nolimits_0^{\pi}
f\big(\Psi(\eta,\vartheta,\varphi)\big)\cdot|\det\Psi'(\eta,\vartheta,\varphi)|\,d\eta\,d\vartheta\,d\varphi
\end{equation}
is left-invariant, hence must coincide with the Haar integral on $\calSU(2)$.
Items~(a) and (b) thus follow by polynomial-time reduction to Euclidean/Riemann integration 
according to Equation~(\ref{e:Integral}).
And Item~(c) follows by polynomial-time embedding
\begin{gather*}
\calC\big(\calU(1)\big) \;\ni\; f \;\mapsto\; \tilde f \;\in\;
\calC(\IH_1)\;\cong\;\calC\big(\calSU(2)\big) , \quad\text{where}\\
\tilde f \;:\; \IH_1\;\ni\; \alpha+i\beta+j\gamma+k\delta \mapsto f\big((\alpha+i\beta)/\sqrt{\alpha^2+\beta^2}\big)\cdot \sqrt{\alpha^2+\beta^2} \;\in\; \IR
\enspace .
\end{gather*}
Since continuous $f$ on compact $\IH_1$ is bounded,
$f\big((\alpha+i\beta)/\sqrt{\alpha^2+\beta^2}\big)\cdot \sqrt{\alpha^2+\beta^2}\to0$
as $\alpha^2+\beta^2\searrow0$. Hence
$\tilde f$ is indeed
well-defined and remains polynomial-time computable
by continuous extension also for $\alpha^2+\beta^2=0$.

\medskip
$\calSO(3)$ is doubly-covered by $\IH_1$, identifying $q\in\IH_1$ with special orthogonal linear map
\[ \IR^3 \;\ni\; (\beta,\gamma,\delta) \;\mapsto\; (\beta',\gamma',\delta')  \;:\;  i\beta'+j\gamma'+k\delta' \overset{!}{=} \big( q\cdot (i\beta+j\gamma+k\delta)\cdot q^{-1} \big) \enspace . \]
Moreover $\calO(3)\cong\calSO(3)\times\{\pm1\}$ and $\calU(2)\cong\calSU(2)\times\calU(1)$.
\qed

\subsection{Implementation and Evaluation}\label{ss:Implement}
Based on the above reduction to ordinary Riemann integration,
we have implemented integration on $\calSU(2)$ in the \texttt{iRRAM} C++ library
\cite{Mue01}. The source code is available at \href{http://github.com/realcomputation/irramplus/tree/master/HAAR}{\texttt{https://github.com/realcomputation/irramplus/tree/master/HAAR}}.
Its empirical evaluation produced the following timing results:
\begin{figure}[H]
    \label{fig:execution-time}
        \centering
            \begin{tikzpicture}
            \begin{axis}[
            xlabel=Precision $n$ referring to absolute error $\leq2^{-n}$,
            ylabel=Logarithm of runtime in seconds,
            grid=major,
            ]
            \addplot[only marks] table [row sep=\\] {
                -13    -4.305006802\\
                -12    -2.520853892\\
                -11    -1.516307923\\
                -10    -0.726900478\\
                -9    0.978605678\\
                -8    2.860749505\\
                -7    3.826900284\\
                -6    6.534514553\\
                -5    9.252760978\\
            };
            \addplot[no markers, red] table [y={create col/linear regression={y=y}}, row sep=\\] {
                x y\\
                -13    -4.305006802\\
                -12    -2.520853892\\
                -11    -1.516307923\\
                -10    -0.726900478\\
                -9    0.978605678\\
                -8    2.860749505\\
                -7    3.826900284\\
                -6    6.534514553\\
                -5    9.252760978\\
            };
            \addplot[no markers, red, domain=-5:0]{\pgfplotstableregressiona*x+\pgfplotstableregressionb};
            \end{axis}
            \end{tikzpicture}
        \caption[]{ Empirical evaluation of rigorous integration on $\calSU(2)$.
        }
    \end{figure}
    Specifically, we chose the Lipschitz-continuous
    function $\IH_1 \ni w+xi+yj+zk \mapsto |w|+|x|+|y|+|z|\in\IR$
    to integrate with\emph{out} letting the algorithm exploit
    its particular symbolic form and symmetry.
    Time measurements were performed on the virtual machine that has Ubuntu 64-bit with 4 cores and 8GB RAM by VMware Workstation 15 Player. The underlying computer has Intel(R) Core(TM) i7-7700K CPU 4.20GHz and 16GB RAM.
    Execution time for each precision is the average execution time of 5 executions.

    Note that the y-axis records the \emph{logarithm} of the execution time in seconds.
    This time is confirmed to grow exponentially with the output precision parameter $n$:
    as expected for a $\classSharpP_1$-complete problem.

\section{Conclusion and Future Work}
\label{sec:Conclusion}
We have devised a computable version of Haar's Theorem:
proven once using the elegant synthetic (implicit) approach
and once developing and analyzing an explicit, imperative algorithm.
And we have established the computational complexity of the Haar integral
to characterize $\classSharpP_1$ for each of the compact groups
$\calU(1),\calU(2),\calO(2),\calO(3),\calSU(2),\calSO(3)$. Moreover, we implemented the algorithm for $\calSU(2)$ in \emph{Exact Real Computation} \cite{DBLP:journals/corr/MullerPP016} and confirmed that the experiment coincides with the complexity theorem.
In fact, our proof shows them mutually second-order polynomial-time Weihrauch reducible
\cite{DBLP:journals/toct/KawamuraC12}.

Future work will generalize the above complexity considerations
to $\calSO(4)$, to $\calSO(d)$, and to further classes of compact metric groups;
and improve the implementation to achieve practical performance.

On the abstract side of our work, an immediate question is whether we can generalize from compact groups to locally compact groups (as was done for the classical Haar's theorem). The price to pay for this generalization in the classic setting is that we no longer obtain a unique probability measure, but merely a locally finite measure identified up to a constant scaling factor. A notion of effective local compactness is available (see \cite{pauly-locallycompact}), but any such generalization seems to require new proof techniques beyond those employed in this article. Recently, Davorin Le\v{s}nik has shown that this one can be done in synthetic topology, provided that one is willing to relax the requirement that measures take values in the lower reals to values in the Borel reals \cite{lesnik-ccc}.

\bibliography{references}

\end{document}